\documentclass[a4paper]{amsart}
\usepackage[utf8]{inputenc}
\usepackage[table]{xcolor} 
\usepackage{graphicx,tikz,eucal}
\usetikzlibrary{matrix,arrows}
\usepackage{tikz-cd}
\usetikzlibrary{calc}
\usetikzlibrary{decorations.pathmorphing}
\usepackage{amssymb,amsmath,amsthm}
\usepackage{lmodern}
\usepackage[T1]{fontenc}
\usepackage[hidelinks]{hyperref}
\usepackage{mathrsfs}
\usepackage{wrapfig}
\usepackage{caption}

\usepackage{thmtools}
\declaretheorem[name=Theorem, parent=section]{thm}

\declaretheorem[name=Proposition, sibling=thm]{prop}

\declaretheorem[name=Example, style=definition, qed=$\triangleleft$, sibling=thm]{example}
\declaretheorem[name=Remark, style=definition, qed=$\triangleleft$, sibling=thm]{rem}

\tikzset{ 
table/.style={
  matrix of nodes,
  row sep=-\pgflinewidth,
  column sep=-\pgflinewidth,
  nodes={rectangle,draw=black,text width=10ex,align=center},
  text depth=0.25ex,
  text height=3ex,
  nodes in empty cells
  },
texto/.style={font=\footnotesize\sffamily},
title/.style={font=\small\sffamily}
}
\definecolor{mblack}{gray}{0.25}
\definecolor{mgray}{gray}{0.7}
\definecolor{mwhite}{gray}{0.95}
\newcommand\nadpis[2]{\node[anchor=base] at ([yshift=4ex]A-1-#1) {#2};}
\newcommand\nadnadpis[2]{\node[anchor=base] at ([yshift=6.5ex]A-1-#1) {#2};}
\newcommand\podpis[2]{\node[anchor=base] at ([yshift=-5.25ex]A-#1) {#2};}
\newcommand\vpis[2]{\node[anchor=base] at ([yshift=-.8ex]A-#1) {#2};}
\newcommand\vpisw[2]{\node[anchor=base] at ([yshift=-.8ex]A-#1) {\textcolor{white}{#2}};}
\newcommand\predpis[2]{\node at ([xshift=-11ex]A-#1-1) {#2};}
\newcommand\ocisluj[1]{\foreach[evaluate={\y=int(\x-1)}] \x in {1,...,#1} {\node at ([xshift=-6ex]A-\x-1) {$\y$};}}
\newcommand{\wigglyampl}{0.5mm}
\newcommand{\wigglyseglen}{2mm}
\newcommand{\wigglywidth}{2.5mm}
\newcommand{\wigglycellheight}{10mm}

\newcommand{\mf}{\mathfrak}
\newcommand{\mc}{\mathcal}
\newcommand{\on}{\operatorname}

\newcommand{\g}{\mathfrak{g}}
\newcommand{\kk}{\mathfrak{k}}

\newcommand{\h}{\mathfrak{h}}

\newcommand{\bw}{{\textstyle\bigwedge}}

\newcommand{\la}{\langle}
\newcommand{\ra}{\rangle}
\newcommand{\R}{\mathbb{R}}

\newcommand{\half}{\frac{1}{2}}

\newcommand{\subsetsep}{\;\subset\;}

\title{A non-abelian duality for (higher) gauge theories}
\author{Ján Pulmann}
\author{Pavol \v{S}evera}
\author{Fridrich Valach}
\address{Section of Mathematics, Universit\'{e} de Gen\`{e}ve, Geneva, Switzerland}
\email{jan.pulmann@unige.ch}
\address{Section of Mathematics, Universit\'{e} de Gen\`{e}ve, Geneva, Switzerland}
\email{pavol.severa@gmail.com}
\address{Mathematical Institute, Faculty of Mathematics and Physics, Charles University, Prague, Czech Republic}
\email{fridrich.valach@gmail.com}
\thanks{Supported in part by  the grant MODFLAT of the European Research Council and the NCCR SwissMAP of the Swiss National Science Foundation.  F.V.\ was supported also by the GAČR Grant EXPRO 19-28628X}

\begin{document}
\maketitle

\begin{abstract}
We consider a TFT on the product of a manifold with an interval, together with a topological and a non-topological boundary condition imposed at the two respective ends. The resulting (in general higher gauge) field theory is non-topological, with different choices of the topological conditions leading to field theories dual to each other. In particular, we recover the electric-magnetic duality, the Poisson-Lie T-duality, and we obtain new higher analogues thereof. 
\end{abstract}

\section{Introduction}

T-duality of 2-dimensional $\sigma$-models has two quite different generalizations: In 2 dimensions it has a non-abelian version called Poisson-Lie T-duality \cite{KS}. In 4 dimensions it has an analogue in electric-magnetic duality. A natural question is whether there is a single mechanism explaining both of these generalizations, giving rise to new dualities of possibly higher gauge theories.

The fact that duality in higher dimensions involves higher gauge theories can be illustrated by the following simple example. If $\Sigma$ is an $n$-dimensional space-time and $p+q+2=n$ ($p,q\geq0$), let us consider the action functional $S(A)=\int_\Sigma F\wedge*F$, where $F=dA$ and $A\in\Omega^p(\Sigma)$. The equations of motion are $d\,{*F}=0$ and (identically) $dF=0$. This pure higher electromagnetism has gauge symmetries $A\mapsto A+dA'$ ($A'\in\Omega^{p-1}(\Sigma)$), gauge symmetries of gauge symmetries $A'\mapsto A'+d A''$, etc. If we introduce $\tilde A\in\Omega^q(\Sigma)$ such that $*F=d\tilde A$ then (at least naively) we get a duality exchanging $p$ and $q$. Even if we start with a model with  no gauge symmetries ($p=0$) or with ordinary gauge symmetries ($p=1$), the dual model will have higher gauge symmetries, if the dimension $n$ is large enough. The general construction we're looking for thus needs to include higher gauge theories.

Inspired by \cite{SCS}, where Poisson-Lie T-duality is explained in terms of suitable boundary conditions of Chern-Simons theory, we propose the following picture for such higher dualities. Suppose $\alpha$ is an $n+1$-dim topological field theory (TFT), $F$ is a non-topological boundary field theory of $\alpha$ requiring a Riemannian metric or a similar structure on the boundary (see \cite{FT}), and  $L$ is a \emph{topological} boundary  condition of $\alpha$. This data gives us a \emph{non-topological $n$-dim field theory}: given an $n$-dimensional Riemannian manifold $\Sigma$, we obtain it from $\alpha$ on $\Sigma\times[0,1]$ with $F$ on $\Sigma\times\{0\}$ and $L$ on $\Sigma\times\{1\}$, which we interpret as a field theory on $\Sigma$.
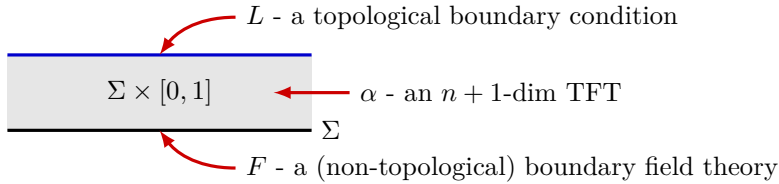
$$
\begin{tikzpicture}[yscale=0.5]
\fill[black!10!white] (0,0)--(0,2)--(4,2)--(4,0);
\draw[very thick](0,0)--(4,0);
\draw[very thick,blue!80!black] (0,2)--(4,2);
\node at (2,1) {$\Sigma\times [0,1]$};
\node[right] at (4,0) {$\Sigma$};
\draw[red!80!black,very thick, <-, >=latex] (3.5,1)--(4.5,1) node[right,black]{$\alpha$ - an $n+1$-dim TFT};
\draw[red!80!black,very thick, <-, >=latex] (2,2) to [out=70,in=180] (3,3) node[right,black]{$L$ - a topological boundary condition};
\draw[red!80!black,very thick, <-, >=latex] (2,0) to [out=-70,in=180] (3,-1) node[right,black]{$F$ - a (non-topological) boundary field theory};
\end{tikzpicture}
$$

If in place of $L$ we use another topological boundary condition $L'$ then we get a possibly different, though closely related $n$-dim field theory. It may easily happen that $L$ and $L'$ are equal (or isomorphic) even though they come from two different \emph{classical} boundary conditions. This is precisely what happens in T-duality, when $\alpha$ is an abelian Chern-Simons theory given by a suitable torus, and $L$ and $L'$ come from two different Lagrangian subgroups of the torus, yet giving the same boundary condition at the quantum level \cite{KaSa}. A suitable choice of $F$, together with $\alpha$ and with $L$ or $L'$, then produces two 2-dim $\sigma$-models with the worldsheet $\Sigma$, linked by T-duality. 

In general we shall call two theories obtained from \emph{the same $\alpha$ and $F$}, but from possibly \emph{different $L$ and $L'$}, dual to each other. The full relation between such theories should come from the $n$-category structure of the set of all topological boundary conditions \cite{Kapustin}.

The aim of our paper is to look at this picture from the BV perspective. The  TFTs we shall consider are of the AKSZ type \cite{AKSZ}, i.e.\ given by a dg symplectic manifold $X$, and $L$'s by dg Lagrangian submanifolds of $X$. $F$'s will be given by suitable data yielding dg Lagrangian submanifolds in $\on{Maps}(T[1]\Sigma,X)$. As we shall see, if we keep $X$ and $F$ fixed, different $L$'s will give rise to quite different (higher) gauge theories, which are dual to each other according to our definition. The (higher) gauge symmetries will appear automatically via the BV formalism.

The examples we obtain include Poisson-Lie T-duality (when $\alpha$ is a Chern-Simons theory and $X=\g[1]$, where $\g$ is the corresponding Lie algebra), electric-magnetic duality ($n=4$, $X=\R^2[2]$), and also many exotic-looking (higher) gauge theories. While any gauge theory, including Yang-Mills, can be put into the $(\alpha,F,L)$ form, we do not find a duality involving Yang-Mills theory due to lack of a suitable $L'$ (more precisely, due to the acyclicity of $X$). Hopefully this issue can be addressed by including supersymmetry into the formalism to get a duality of the Montonen-Olive type \cite{MO}, but we leave it to a future work.

The main technical tool that we use for calculations is the derived intersection of Lagrangian submanifolds (or maps) introduced in \cite{ptvv}. It is a somewhat complicated concept, so we use its simplified version which we describe in some detail. The price to pay for this simplification is that our calculations are sometimes only local (since some relevant objects may exist only locally). In fact, to get a truly global description, we would need to use, as in \cite{ptvv}, higher derived stacks in place of dg manifolds. On the other hand, our simplified methods give relatively simple action functionals in the BV formalism, and globalization can often be done ad hoc. We leave these global issues for a future work as well.

\subsection*{Acknowledgments}
We would like to thank to Branislav Jur\v co for useful comments on a preliminary version of this paper.

\subsection*{Notation and terminology}
If $V$ is a $\mathbb Z$-graded vector space and if $n\in\mathbb Z$ then $V[n]$ denotes $V$ with the grading shifted by $n$: $V[n]_k=V_{n+k}$.

A \emph{graded manifold} is a supermanifold $X$ whose algebra of functions is $\mathbb Z$-graded (and not just $\mathbb Z_2$-graded). An \emph{N-manifold} (non-negatively graded manifold) corresponds to the case of a $\mathbb Z_{\geq0}$-grading; equivalently, it is a supermanifold with an action of the semigroup $(\R,\times)$ such that $-1\in\R$ acts as the parity operator.

For any graded manifold $X$ let $E_X$ denote its \emph{Euler vector field} given by $E_Xf=(\deg f)f$.

A \emph{differential graded (dg) manifold} is a graded manifold $X$ equipped with a vector field (differential) $Q_X$ such that $\deg Q_X=1$ and $Q_X^2=0$. An NQ-manifold is an N-manifold with a differential.

The simplest examples of NQ-manifolds are $T[1]M$, where $M$ is an ordinary manifold. We have $C^\infty(T[1]M)=\Omega(M)$ and $Q_{T[1]M}$ is the de Rham differential. If $M$ is oriented then on $T[1]M$ we have a natural $Q_{T[1]M}$-invariant volume form, corresponding to the integration of differential forms on $M$.

 If $\g$ is a Lie algebra then $\g[1]$ is another example of an NQ-manifold, with $C^\infty(\g[1])=\bw\g^*$, and with $Q_{\g[1]}$ being the Chevalley-Eilenberg differential.

A \emph{dg symplectic manifold} is a dg manifold $X$ equipped with a symplectic form $\omega$ s.t.\ $L_{Q_X}\omega=0$ and $\deg\omega=n$ for some $n\in\mathbb Z$. If $n\neq0$ the 1-form $\theta=i_{E_X}\omega/n$ satisfies $d\theta=\omega$. If $n\neq -1$ then $H_X=i_{Q_X}i_{E_X}\omega/(n+1)$ is a Hamiltonian of the vector field $Q_X$.

If $n=-1$ and if a Hamiltonian $H_X$ (of degree 0) of $Q_X$ is given then $X$ is a \emph{classical BV manifold}. In other words, $X$ is a graded manifold with a symplectic form $\omega$ of degree $-1$ and with a function $H_X$ of degree 0, satisfying the \emph{classical master equation} (CME)
$$\{H_X,H_X\}=0.$$
Typically $X$ will be infinite-dimensional (a space of fields) and the Hamiltonian $H_X$ will be rather denoted by $S_X$, playing the role of an action functional.

\section{AKSZ models}
Let us review the construction of the AKSZ model \cite{AKSZ}. It is given by a symplectic NQ manifold $(X,Q_X,\omega_X)$  with $\deg\omega_X=n$. If $M$ is an oriented compact $n+1$ dimensional manifold then 
$$\mc M:=\on{Maps}(T[1]M,X)$$
is an infinite-dimensional classical BV manifold with the symplectic form
 $$\omega_\mc M (u,v)=\int_{T[1]M}\omega_X(u,v).$$ 
The differential $Q_\mc M$ is given by the difference of $Q_X$ and $Q_{T[1]M}$.

The Hamiltonian $S_\mc M$  of the homological vector field $Q_\mc M$ can be computed as
\begin{equation}\label{akszion}
S_\mc M(f):=\int_{T[1]M} i_{d_M} f^*\theta_X + f^*H_X
\end{equation}
where $\theta_X\in\Omega^1(X)$ is an arbitrary 1-form such that $d\theta_X=\omega_X$ (e.g.\ $\theta_X=i_{E_X}\omega_X/n$ where $E_X$ is the Euler vector field) and $d_M=Q_{T[1]M}$ is the de Rham differential on $M$. $S_\mc M$ is the action functional of the AKSZ model (the akszion functional). The critical points of $S_\mc M$ are the $Q$-preserving maps $f:T[1]M\to X$. By construction $S_\mc M$ solves the classical master equation (CME)
$$\{S_\mc M,S_\mc M\}=0.$$

\begin{example}[Chern-Simons theory]\label{ex:CS}
If $\g$ is a Lie algebra with an invariant inner product $\la,\ra$ then $X=\g[1]$ has a $Q_X$-invariant symplectic form $\omega_X$ (given by $\la,\ra$) of degree $n=2$. We have $\mc M=\Omega(M,\g)[1]$ and the action $S_\mc M$ is
$$S_\mc M(A)=\int_M\frac{1}{2}\la A,dA\ra+\frac{1}{6}\la[A,A],A\ra.$$
If $A^{(i)}$ is the $i$-form part of $A$ ($i=0,1,2,3$) then $A^{(1)}$ is a connection (a field), $A^{(0)}$ the ghost corresponding to the gauge transformations, $A^{(2)}$ is the antifield of $A^{(1)}$ and $A^{(3)}$ the antifield of $A^{(0)}$.
\end{example}

\section{Topological boundary conditions of AKSZ models}
If $M$ has a boundary we can still define $\omega_\mc M$ as above, but it is no longer $Q$-invariant. Namely, the space of boundary fields $\mc X:=\on{Maps}(T[1]\partial M,X)$ is dg symplectic with $\deg\omega_\mc X=0$, where
 $$\omega_\mc X (u,v)=\int_{T[1]\partial M}\omega_X(u,v),$$ 
 and Stokes theorem gives us
\begin{equation}\label{LQomega}
L_{Q_\mc M}\omega_\mc M=b^*\omega_\mc X
\end{equation}
where $b:\mc M\to\mc X$ is the restriction map.

A dg Lagrangian submanifold $L\subset X$ will then give us a boundary condition for the AKSZ model:  we restrict $\mc M$ to 
$$\mc M_L:=b^{-1}\bigl( \on{Maps}(T[1]\partial M,L) \bigr) \subsetsep \mc M,$$
which is (unlike $\mc M$) a classical BV manifold. The corresponding action functional 
(i.e.\ the Hamiltonian of $Q_{\mc M_L}$)  is still
$$S_{\mc M_L}(f)=\int_{T[1]M} i_{d_M} f^*\theta_X + f^*H_X$$
provided $\theta_X=i_{E_X}\omega_X/n$ (as then $\theta_X|_L=0$); other choices of $\theta_X$ (s.t.\ $d\theta_X=\omega_X$)  produce boundary terms in $S_{\mc M_L}$.

\begin{example}[Topological boundary conditions in Chern-Simons theory]\label{ex:pltop}
Continuing Example \ref{ex:CS}, if $\h\subset\g$ is a Lagrangian Lie subalgebra then
 $$L:=\h[1]\subsetsep\g[1]=:X$$
 is a dg Lagrangian submanifold. The corresponding boundary condition of the Chern-Simons theory requires $A|_{\partial M}$ to have values in $\h$.
\end{example}

More generally one can use Lagrangian maps (see Section \ref{sec:lagmaps}) as topological boundary conditions.

\section{Non-topological boundary conditions of AKSZ models}\label{sec:non_top}
Suppose now that $\Sigma=\partial M$ is endowed with a Riemannian metric or with a similar geometric structure. A non-topological boundary condition of the AKSZ model should be a dg Lagrangian submanifold in the space of boundary fields
$$\mc F\subsetsep \mc X:=\on{Maps}(T[1]\Sigma,X)$$
depending on the choice of the Riemannian metric.

 The space of fields satisfying the boundary condition
$$\mc M_\mc F=\mc M\times_\mc X\mc F=\bigl\{f\in\on{Maps}(T[1]M,X) \bigm| f|_{T[1]\Sigma}\in\mc F\bigr\}$$
is again a classical BV manifold.

To compute $S_{\mc M_\mc F}$ we need to suppose that $\mc F\subset\mc X$ is \emph{exact} Lagrangian. Namely, if we define the 1-form $\theta_\mc X$ on $\mc X$ via 
$$\theta_\mc X(u)=\int_{T[1]\Sigma}\theta_X(u)$$
so that $d\theta_\mc X=\omega_\mc X$, then we need to have a functional $S_\mc F$ on the space of fields $\mc F$ such that 
$$d S_\mc F=\theta_\mc X|_\mc F.$$
 Then
$$S_{\mc M_\mc F}=S_\mc M+S_\mc F$$
which follows from $d S_\mc M=i_Q\omega_\mc M-\theta_\mc X$. 

There is a particularly simple class of dg Lagrangian submanifolds of $\mc X$: if $\mc F\subset\mc X$ is a graded Lagrangian submanifold such that $C^\infty(\mc F)$ is non-positively graded (i.e.\ $\mc F$ is modeled by a non-negatively graded vector space) then $\mc F$ is automatically dg Lagrangian for degree reasons. We shall call this type of $\mc F$'s \emph{ghostless}. 

\begin{rem}
Let us suppose that $\theta_X=i_{E_X}\omega_X/n$. Then $\theta_\mc X|_\mc F=0$, and thus we can take $S_\mc F=0$, iff $\mc F$ is $E_X$-invariant, i.e.\ iff it is $E_{T[1]\Sigma}$-invariant. In other words we have $S_\mc F=0$ for scale-invariant boundary conditions.
\end{rem}

\begin{example}\label{ex:pl1}
In the case of Chern-Simons theory ($n=2$, $X=\g[1]$) a natural boundary condition is given by a \emph{generalized metric} on $\g$, i.e.\ by a symmetric linear map $E:\g\to\g$ such that $E^2=1$ (a reflection). If $\Sigma$ has  a pseudo-conformal  structure then
$$\mc F=\{A\in\Omega^1(\Sigma,\g)\mid *A=EA\}\oplus\Omega^2(\Sigma,\g)\subsetsep\Omega(\Sigma,\g)=\mc X$$
is a dg Lagrangian submanifold. This $\mc F$ is scale-invariant and ghostless.
\end{example}

\section{AKSZ sandwich and the duality, or plurality, of field theories}\label{sec:duality}

Let $\Sigma$ be an oriented $n$-dimensional manifold.
Let us consider the AKSZ model on $M=\Sigma\times I$ (where $I=[0,1]$), with a non-topological boundary condition $\mc F$ on $\Sigma\times\{0\}$ and with a topological boundary condition given by $L\subset X$
 on $\Sigma\times\{1\}$. The resulting action functional is again a solution of the CME, and can be interpreted as a classical field theory on $\Sigma$ (with infinitely many fields due to their dependence on $I$; see \S\ref{sec:computing} for an equivalent model with finitely many fields). We shall call this model an \emph{AKSZ sandwich}.

Let us  fix $X$ and $\mc F$ and consider AKSZ sandwiches with different $L$'s. While these field theories are not entirely equivalent, their difference is, in some sense, purely topological. We will call them \emph{dual} to each other, since this construction contains as a special case the Poisson-Lie T-duality  and the electric-magnetic duality, as we show below. For example, Poisson-Lie T-duality is given by Chern-Simons theory with the boundary conditions from examples \ref{ex:pltop} and \ref{ex:pl1} -- the corresponding sandwich model is equivalent to a 2-dim $\sigma$-model with the target $G/H$ and the duality itself corresponds to different choices of $\h\subset\g$.

The ``true'' duality should correspond to the cases when, at the quantum level, the topological boundary conditions given by the different $L$'s coincide (or are isomorphic).
 The word ``plurality'', suggested by R.\ von Unge \cite{vU} in the context of Poisson-Lie T-duality, might be more appropriate, as there may be more than 2 suitable $L$'s.

\medskip

Let us now explain how to reduce the AKSZ sandwich model to an equivalent field theory on $\Sigma$ with \emph{finitely} many fields. In some sense, this corresponds to integrating out superfluous fields (as done in the case of Poisson-Lie T-duality in \cite{SCS}). Our method is described in the following two sections, but other (and possibly more enlightening) approaches may exist.

\section{Resolutions of dg Lagrangian submanifolds}\label{sec:lagmaps}

This section contains some preliminaries needed for calculations with the AKSZ sandwiches. We start by defining a ``baby version'' of a \emph{Lagrangian map} $\ell\colon L\to X$, introduced in \cite{ptvv}, which is a more flexible notion than a dg Lagrangian submanifold $L\subset X$. 

First, we shall say that an NQ-manifold $Y$ is \emph{acyclic} if it is isomorphic to $T[1]Z$ for some N-manifold $Z$, or equivalently, if its tangential cohomology vanishes everywhere. 

Suppose now that $X$ and $Y$ are NQ symplectic, $\deg\omega_X=\deg\omega_Y=n$, and that $Y$ is  acyclic. If $L$ is a dg Lagrangian relation between $X$ and $Y$, i.e.\ if $L\subset X\times\bar Y$ is a dg Lagrangian submanifold, then we want to see $L$ as a ``generalized dg Lagrangian submanifold'' of $X$ (the ideology behind is that $Y$ is seen as equivalent to a point). We therefore call the projection $\ell:=p_X: L\to X$ a \emph{Lagrangian map}. As a particular example, $\on{id}\colon Y\to Y$ is Lagrangian (here $L\subset Y\times\bar Y$ is the diagonal embedding of the acyclic $Y$).

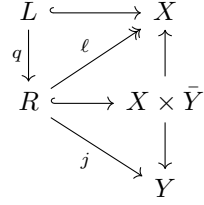
\begin{wrapfigure}[11]{r}{3.8cm}
	\centering
	\begin{tikzcd}
	L \arrow[r, hook] \arrow[d, "q"']                       & X                                  \\
	R \arrow[r, hook] \arrow[ru, "\ell", two heads] \arrow[rd, "j"'] & X\times\bar{Y} \arrow[u] \arrow[d] \\
	& Y                                
	\end{tikzcd}

	\caption{Resolution of $L\subset X$ via $R$.}
\end{wrapfigure}

By a \emph{resolution} of a dg Lagrangian submanifold $L\subset X$ we mean a Lagrangian map $\ell\colon R \to X$, which is a surjective submersion, together with a quasiisomorphism $q\colon L \to R$ such that $\ell \circ q$ coincides with the inclusion $L\to X$.

Let us notice that in this case the projection $j\colon R\to Y$ is an immersion. We can thus view $R$ as a coisotropic submanifold of $Y$ and $X$ as the symplectic reduction (the space of null leaves) of $R$.

\begin{example}\label{ex:G/H}
Let $\g$ be a non-positively graded Lie algebra with an invariant pairing turning $X=\g[1]$ to a degree $n$ symplectic NQ manifold. Let $\h\subset\g$ be a Lagrangian graded Lie subalgebra, so that $\h[1]\subset\g[1]$ is a dg Lagrangian submanifold. We can get its resolution as follows.

 If $H\subset G$ are NQ-groups (with zero differential) integrating $\h\subset\g$ then 
$$Y:=\g[1]\times T^*[n](G/H)$$
 with the Hamiltonian $H_Y=H_{\g[1]}+H_\text{action}$ is a degree $n$ symplectic acyclic NQ manifold. Here $H_\text{action}$ is the action of the Lie algebra $\g$ on $G/H$ (a map from $\g$ to vector fields on $G/H$), seen as a function on $\g[1]\times T^*[n](G/H)$.

The NQ submanifold $R:=\g[1]\times G/H\subset Y$ is coisotropic, and the projection $p_{\g[1]}\colon R\to\g[1]$ is its symplectic reduction. The inclusion $\h[1]\times 1\subset R$ is a quasiisomorphism. This makes $R$ a resolution of $\h[1]\subset\g[1]$.
\end{example}

\section{Computing the sandwich action functional}\label{sec:computing}

Seeing the AKSZ sandwich model from \S\ref{sec:duality} as a field theory on  $\Sigma$ gives us an infinite number of fields due to their dependence on  $I$. There is, however, a way to get an equivalent model with a \emph{finite} number of fields.

In the language of  \cite{ptvv}, the AKSZ sandwich model  can be seen as the derived intersection of the dg Lagrangian submanifolds  $\mc F$ and $\mc L$ of $\mc X$, where
$$\mc L=\on{Maps}(T[1]\Sigma,L).$$
This means the following. We replace the dg Lagrangian submanifold $\mc L\subset\mc X$ by a resolution $\lambda:\mc R\to\mc X$. The derived intersection of $\mc L$ with $\mc F\subset\mc X$ is then
\begin{equation}\label{derint}
\mc R_\mc F:=\lambda^{-1}(\mc F)\subset\mc R
\end{equation}
and, according to \cite{ptvv}, it is (up to homotopy) a classical BV manifold. Indeed, the AKSZ model gives us a particular $\mc R$, namely 
$$\mc R=\bigl\{f:T[1](\Sigma\times I)\to X \bigm| f|_{T[1](\Sigma\times\{1\})}\in\mc L\bigr\},$$
$\lambda:\mc R\to\mc X$ is given by $f\mapsto f|_{T[1](\Sigma\times\{0\})}$, 
and $\mc R_\mc F$ is then the space of fields of the AKSZ sandwich.

To find a resolution $\mc R$ of $\mc L$ it is enough to find a resolution $\ell:R\to X$ of $L\subset X$, and then set
$$\mc R=\on{Maps}(T[1]\Sigma,R).$$
Again, the AKSZ sandwich gives us a particular resolution $R$, namely the path space
$$R=\{f:T[1]I\to X\mid f(1)\in L\}.$$
The idea is to use a \emph{finite-dimensional} $R$ instead, which gives a more manageable but quasi-isomorphic space of fields.

%

Let us now describe  how the space of fields $\mc R_\mc F$ (or rather its reduction) becomes a classical BV manifold and how to write down the action functional, following (a baby version of) \cite{ptvv}.

The differential on $\mc R=\on{Maps}(T[1]\Sigma,R)$ is, as usual, given by the difference of the differentials on $R$ and on $T[1]\Sigma$. The space of fields $\mc R_\mc F\subset\mc R$ is a dg submanifold.

To get the BV  2-form on $\mc R_\mc F$ we need to write the acyclic symplectic $Y$ in the form
\begin{equation}\label{tildeY}
Y=T^*[n]\tilde Y
\end{equation}
where $\tilde Y$ is another symplectic NQ manifold with $\deg\omega_{\tilde Y}=n-1$.
 The symplectic form on $Y$ comes from $T^*[n]$ and the Hamiltonian  $H_Y$ is the Poisson structure of $\tilde Y$ (a quadratic function on $T^*[n]\tilde Y$) plus $Q_{\tilde Y}$ (a linear function on $T^*[n]\tilde Y$).%
\footnote{The symplectic form $\omega_{\tilde Y}$ gives us an isomorphism of graded manifolds $T^*[n]\tilde Y\cong T[1]\tilde Y$. The resulting differential on $T[1]\tilde Y$ is the standard one; this, in particular, implies the acyclicity of $T^*[n]\tilde Y$.} 
 Every acyclic symplectic $Y$ is locally of this form (we can even demand $Q_{\tilde Y}=0$), but globally there may be an obstruction in $H^{n+1}(Y^0;\R)$ (where $Y^0$ is the degree-0 part of $Y$).

Let
$$\ell\colon R\to X,\quad j\colon R\to Y,\quad p:Y\to\tilde Y$$
 be the projections (recall that $\ell$ is a surjective submersion and $j$  a coisotropic immersion). The sought-after 2-form $\tilde\omega$ on $\mc R_\mc F$ is 
 $$\tilde\omega(u,v)=\int_{T[1]\Sigma}(p\circ j)^*\omega_{\tilde Y}(u,v).$$
 As needed, it is closed,
 of degree $-1$, and satisfies $L_Q\tilde\omega=0$, but it may be degenerate.  Let us suppose that the space of the null leaves of $\tilde\omega$ is a graded manifold 
$$\mc Z:=(\mc R_\mc F)^\mathrm{reduced},$$
 i.e.\ that we have a surjective submersion ${\mc R_\mc F}\to\mc Z$ whose fibers are the null leaves of the 2-form. Then both $\tilde\omega$ and the differential descend to $\mc Z$ and make it to a classical BV manifold 
$$(\mc Z,\tilde\omega_\mc Z,Q_\mc Z).$$
 By construction this BV space of fields is equivalent to the AKSZ sandwich given by $\mc L$ and $\mc F$.
 
The action functional, i.e.\ the Hamiltonian generating $Q_\mc Z$, can be computed as follows. Let $H_\text{rel}\in C^\infty(R)$ be a function such that
$$d H_\text{rel}=j^*\theta_Y^\text{taut}-\ell^*\theta_X$$
where $\theta_Y^\text{taut}$ is the tautological 1-form on $Y=T^*[n]\tilde Y$.
Then we have
\begin{equation}
\label{eq:final}
S_{\mc Z}([f])=S_{\mathrm{AKSZ}}^{\tilde Y}(p\circ j\circ f)+S_\mc F(\ell\circ f)-\int_{T[1]\Sigma}f^* H_\text{rel}
\end{equation}
where $p\colon Y\to \tilde Y$ is the projection and $[f]\in\mc Z$ denotes the class of $f\in\mc R_\mc F$, and $S_{\mathrm{AKSZ}}^{\tilde Y}$ is the AKSZ action functional given by the symplectic NQ-manifold $\tilde Y$. Recall that $S_\mc F$ is a functional on $\mc F$ such that $d S_\mc F=\theta_\mc X|_\mc F$.
The function $H_\text{rel}$ can be computed as
$$H_\text{rel}=\tfrac1n i_{E_R}(j^*\theta^{\mathrm{taut}}_Y-\ell^*\theta_X).$$

\begin{rem}
Our construction can be explained also as follows: we replace the triple $(X,\mc F\subset\mc X,L)$ with $(Y,\mc R_\mc F\subset\mc Y,Y)$ (with $R$ seen as a coisotropic submanifold of $Y$ and $\ell:Y\to Y$ being the identity), where $\mc Y=\on{Maps}(T[1]\Sigma,Y)$. Even if $\mc F$ is ghostless, $\mc R_\mc F$ is often not, which explains the emergence of gauge symmetries in our setup.
\end{rem}

\begin{rem}
In \cite{ptvv} a more flexible definition of symplectic and Lagrangian structures is used, which removes two problems of our approach - the fact that $\tilde Y$ may exist only locally and the requirement that $\mc Z$ is smooth. On the other hand, relating the methods of \emph{op.\ cit.} to the standard BV formalism seems somewhat tricky.
\end{rem}

\begin{example}\label{ex:gradedplgrps}
Continuing Example \ref{ex:G/H}, we can put $Y=\g[1]\times T^*[n](G/H)$ to the form \eqref{tildeY} as follows. Suppose that $\h'\subset\g$ is a Lagrangian Lie subalgebra 
such that $\h\cap\h'=0$. In other words, $(\g,\h,\h')$ is a graded Manin triple, and $H$ and $H'$ thus become graded Poisson-Lie groups. The pairing on $\g$ gives us an isomorphism
$$\h^*\cong\h'[2-n].$$
The inclusion $H'\subset G$ then gives us a local diffeomorphism $H'\to G/H$. Let us suppose that it is a global diffeomorphism (or alternatively work locally in $G/H$).

In this case we can set 
$$\tilde Y= T^*[n-1] H' \cong\h[1]\times H'$$
(with the isomorphism given by the left translation). The Hamiltonian of $Q_{\tilde Y}$ is the Poisson structure on $H'$ seen as a function on $T^*[n-1] H'$. The equality \eqref{tildeY} appears via
\begin{align*}
T^*[n](\h[1]\times H')&=T^*[n]\h[1]\times T^*[n]H'\\
&=(\h'[1]\oplus\h[1])\times T^*[n] H'=\g[1]\times T^*[n]H'.\qedhere
\end{align*}
\end{example}

\section{Boundary conditions and $G$-structures}\label{sec:ul}
Let us now describe in some detail non-topological boundary conditions $\mc F$ of ``ultralocal type'', i.e.\ given by a choice of a graded Lagrangian submanifold 
$$\mc F_p\subsetsep\mc X_p:=\on{Maps}(T_p[1]\Sigma,X)$$
 for every $p\in\Sigma$. (This is done for simplicity; a local $\mc F$, depending on a finite number of derivatives, would be equally good.)  Notice that $\mc X_p$ is finite-dimensional, unlike $\mc X$. 

In more detail, we have a $\bw^n T^*_p\Sigma$-valued symplectic form on $\mc X_p$ given by the Berezin integral
$$\omega_{\mc X_p}(u,v)=\bigl(\omega_X(u,v)\bigr)^\text{top},$$
$\mc X_p$'s are the fibers of a fiber bundle  $\hat{\mc X}\to\Sigma$, and $\mc X=\Gamma(\hat{\mc X})$ is the space of its sections. The Lagrangian graded submanifolds $\mc F_p$ form a subbundle $\hat{\mc F}\subset\hat{\mc X}$, and $\mc F=\Gamma(\hat{\mc F})$.

As in \S \ref{sec:non_top}, we shall call $\mc F_p$ \emph{ghostless} if the algebra $C^\infty(\mc F_p)$ is non-positively graded. If $\mc F_p$ is ghostless for every $p\in\Sigma$ then $\mc F$ is ghostless and thus dg Lagrangian. Otherwise $\mc F$ is just a graded Lagrangian submanifold and we need to find conditions on $\mc F_p$'s to make it a dg submanifold. 

\begin{rem}\label{rem:negF}
A ghostless $\mc F_p$ is uniquely determined by the  Lagrangian submanifold
$$\mc F_p^{\,0}=\mc F_p\cap\mc X_p^{\,0}\subsetsep \mc X_p^{\,0}$$
where $\mc X_p^{\,0}\subset\mc X_p$ is the (finite-dimensional and ordinary, i.e.\ not graded) manifold of grading-preserving maps $T_p[1]\Sigma\to X$. The Lagrangian submanifold $\mc F_p^{\,0}$ can be chosen arbitrarily. This gives rise to an easy class of examples.

(Lagrangian submanifolds $\mc F_p^{\,0}\subset\mc X_p^{\,0}$ were introduced in \cite{some} as ``higher Hamiltonians'' with the purpose of explaining and possibly generalizing Poisson-Lie T-duality. In this sense our paper fulfills the dream of \cite{some}.)
\end{rem}

We now restrict our attention to ultralocal $\mc F$'s coming from $G$-structures.
Let us fix an $n$-dimensional oriented vector space $V$ (a model of $T_p\Sigma$) and suppose that $\Sigma$ is endowed with a $G$-structure for some $G\subset GL_+(V)$ (e.g.\ if $V$ is equipped with an inner product then $G=SO(V)$ would correspond to a Riemannian metric on $\Sigma$). This means that for every $p\in\Sigma$ we have a family of isomorphisms $T_p\Sigma\cong V$ on which $G$ acts transitively. These isomorphisms form a principal $G$-bundle $P\to\Sigma$. If 
$$\mc X_V:=\on{Maps}(V[1],X)$$
then $\hat{\mc X}\to\Sigma$ is the corresponding associated bundle $P\times_G\mc X_V$. On $\mc X_V$ we have a $\bw^nV^*$-valued symplectic form.

 Let 
$$\mc F_V\subset \mc X_V$$
be a $G$-invariant Lagrangian graded submanifold. It can be used to define  $\mc F_p$'s: the isomorphisms $T_p\Sigma\cong V$ (the elements of $P$)  turn $\mc F_V$ to $\mc F_p\subset\mc X_p$ for every $p\in\Sigma$. In other words, $\hat{\mc F}\to\Sigma$ is the  associated subbundle 
$$\hat{\mc F}=P\times_G\mc F_V\subsetsep P\times_G\mc X_V=\hat{\mc X}.$$


\begin{prop}\label{prop:F0}
	$\mc F\subset\mc X$, coming from a $G$-invariant $\mc F_V\subset\mc X_V$, is dg Lagrangian for every $G$-structure on $\Sigma$ iff
	\begin{enumerate}
		\item $\mc F_V\subset\mc X_V$ is dg Lagrangian, where $Q$ on $\mc X_V$ comes from $Q$ on $X$
		\item $T\mc F_V\subset T\mc X_V$ is a $\bw V^*$-submodule (where we use the fact that $T_f\mc X_V=\Gamma(f^*TX)$ is a $C^\infty(V[1])=\bw V^*$-module for every $f\in\mc X_V$)
		\item $\mc F_V\subset\mc X_V$ is invariant under the graded Lie algebra of vector fields on $V[1]$ vanishing at the origin at least quadratically.
	\end{enumerate}
	The last condition is dropped if we consider only $G$-structures which admit a compatible torsion-free connection.
\end{prop}
\begin{proof}
	Locally, $\mc X=\on{Maps}(T[1]\Sigma,X)$ is the graded vector space of forms on $\Sigma$ with 
	values in $X$ (seeing $X$ locally as a graded vector space). This allows us to see the vector field $Q_\mc X$ as a (non-linear and grading-preserving) map $\mc X\to\mc X[1]$.
	
	 Choosing a basis of $V$, a $G$-structure on $\Sigma$ and its local section give a frame $e_1, \dots, e_n$ and we denote the dual 1-forms as $\theta^k$. Let $\sigma^i$ be local coordinates on $\Sigma$ centered at a point $p\in\Sigma$; $\sigma$'s and $\theta$'s are thus local coordinates on $T[1]\Sigma$.
	
	For  
	$$\phi\in\mc F\subsetsep\mc X=\Omega(\Sigma,X)$$
	 we have,%
\footnote{Strictly speaking, as usual when working with ``points'' of supermanifolds, $\phi$ should be allowed to depend on auxiliary odd parameters and its total parity should be even. We suppress these odd parameters in the notation.}
 close to $p\in \Sigma$,
	\[ \phi(\sigma, \theta) = \phi(0, \theta) + \sigma^i \delta_i \phi(0, \theta) + O(|\sigma|^2) \,,\]
	where $\phi(0, \theta)$ is an element of 
	$$\mc F_V \subset \mc X_V=\on{Maps}(V[1],X),$$ and $\delta_i\phi(0, \theta)$ are vectors tangent to $\mc F_V$ at $\phi(0, \theta)$.
	To check that $\mc F$ is preserved by the differential on $\mc X$, we can just check the value of $Q_{\mc X}(\phi)$ at $\sigma=0$:
	\[ (Q_{\mc X}\phi)(0, \theta) = Q_X(\phi(0, \theta)) -  (d_\Sigma \sigma^i) \delta_i \phi_0(0, \theta) - d_\Sigma|_p \phi(0, \theta) \,.\]
In the last term, $d_\Sigma|_p  = -\half\theta^i([e_j, e_k])|_p \theta^j \theta^k\partial_{\theta^i}$.
	Since the $G$-structure can be arbitrary, each of these three terms has to be tangent to $\mc F_V$, giving the three conditions from the proposition.

	If a $G$-structure has a compatible connection, the vanishing of its torsion tells us that the matrices $(C_k)^i_j = \theta^i([e_j, e_k])$ 
	lie in $\on{Lie}(G)$ for every $k$. Since $\mc F_V$ is invariant under $\on{Lie}(G)$ and multiplication by $\theta^k$, the third condition is
	satisfied.
\end{proof}

Finally we should remember that $\mc F$ needs to be \emph{exact} Lagrangian. On $\mc X_p$ we have a $\bw^nT^*_p\Sigma$-valued 1-form $\theta_{\mc X_p}$ given by 
$$\theta_{\mc X_p}(u)=\bigl(\theta_X(u)\bigr)^\text{top}\in\bw^nT^*_p\Sigma$$
satisfying $d\theta_{\mc X_p}=\omega_{\mc X_p}$ and we need $\theta_{\mc X_p}|_{\mc F_p}$ to be exact, equal to $dH_{\mc F_p}$ for some $\bw^nT^*_p\Sigma$-valued function $H_{\mc F_p}$ on $\mc F_p$. Moreover $H_{\mc F_p}$ should depend smoothly on $p\in\Sigma$. We then have
$$S_\mc F=\int_{p\in\Sigma}H_{\mc F_p}.$$
 In the case of $\mc F$ given by a $G$-structure and by $\mc F_V\subset\mc X_V$ we  need a suitable $\bw^nV^*$-valued function $H_{\mc F_V}$ on $\mc F_V$.

\begin{example}[Generalized metric]\label{ex:GM}
Let us repeat Example \ref{ex:pl1} from our current point of view.
We have $n=2$, $X=\g[1]$ for some quadratic Lie algebra $\g$ (i.e.\ the AKSZ model is the Chern-Simons theory given by $\g$).

Let $V$ be an oriented 2-dim vector space equipped with an endomorphism (reflection)
$$*:V\to V,\quad *^2=1,\quad \on{Tr}*=0.$$
 Let $K\subset GL_+(V)$ be the group of linear transformations commuting with $*$. $K$-structures are thus pseudo-conformal structures (giving a Hodge star $*$ on $\Omega^1(\Sigma)$).

We have $\mc X_V=\bw V^*\otimes\g[1]$, and in particular 
$$\mc X_V^{\,0}=V^*\otimes\g=\on{Hom}(V,\g).$$
 If $E:\g\to\g$ is a symmetric linear map such that $E^2=1$ then 
$$\mc F_V^{\,0}:=\{T:V\to\g \mid T*=ET\}\subsetsep\on{Hom}(V,\g)=\mc X_V^{\,0}$$
is a Lagrangian vector subspace. The map $E$ is called \emph{a generalized metric on $\g$}.\footnote{Usually one also adds the requirement that $\la\cdot,E\cdot\ra$ is positive definite, which corresponds to the positivity of the corresponding Hamiltonian.}

Then
$$\mc F_V:=\mc F_V^{\,0}\oplus\bigl(\bw^2V^*\otimes\g[1]\bigr)$$
is ghostless and thus satisfies the conditions of Proposition \ref{prop:F0}. Given a pseudoconformal structure  on $\Sigma$ we  have
\[\mc F=\{A\in\Omega^1(\Sigma,\g)\mid *A=EA\}\oplus\Omega^2(\Sigma,\g)\subsetsep\Omega(\Sigma,\g)=\mc X.\qedhere\]
\end{example}

\begin{example}[Dressing cosets]\label{ex:dressing}
As a small generalization, let $\mf i\subset\g$ be an isotropic Lie  subalgebra and let $E\colon\mf i^\perp/\mf i\to \mf i^\perp/\mf i$ be an $\mf i$-invariant symmetric linear map such that $E^2=1$. Let
$$\mc F_V^{\,0}:=\{T:V\to\mf i^\perp \mid PT*=EPT\}\subsetsep\on{Hom}(V,\g)=\mc X_V^{\,0}$$
where $P:\mf i^\perp\to\mf i^\perp/\mf i$ is the projection. Then
$$\mc F_V:=\mf i[1] \oplus \mc F_V^{\,0}\oplus \bigl(\bw^2V^*\otimes\mf i^\perp[1]\bigr)\subsetsep\mc X_V$$
is again a $K$-invariant graded submanifold satisfying the requirements of Proposition \ref{prop:F0}. The corresponding $\mc F$ is
$$\mc F=\Omega^0(\Sigma,\mf i)\oplus\{A\in\Omega^1(\Sigma,\mf i^\perp)\mid *PA=EPA\}\oplus\Omega^2(\Sigma,\mf i^\perp)\subsetsep\Omega(\Sigma,\g)=\mc X.$$
Note that this example is not ghostless. It appears in the Poisson-Lie T-duality for gauged 2-dim $\sigma$-models \cite{KS2}.
\end{example}

\begin{example}[4-dimensional electromagnetism with several charges]
\label{ex:4dbdry}
Let us consider the case of $n=4$ and $X=W[2]$ for some symplectic vector space $W$. For degree reasons we must have $Q_X=0$. Let $V$ be an oriented 4-dim vector space with a Minkowski inner product, so that on $\bw^2V$ the Hodge operator satisfies $*^2=-1$. Let $K\subset GL_+(V)$ be the group preserving the inner product up to rescaling.

We have $\mc X_V=\bw V^*\otimes W[2]$, in particular $\mc X_V^{\,0}=\on{Hom}(\bw^2 V,W)$. If $J:W\to W$ is a linear map such that $J^2=-1$ and $\omega(u,Jv)=\omega(v,Ju)$ then
$$\mc F_V^{\,0}:=\{T:\bw^2V\to W \mid T*=JT\}\subsetsep\on{Hom}(\bw^2 V,W)=\mc X_V^{\,0}$$
is a $K$-invariant Lagrangian vector subspace and
$$\mc F_V:=\mc F_V^{\,0}\oplus (\bw^{\geq 3}V^*\otimes W[2])\subsetsep \mc X_V$$
is ghostless. 

It is natural (for Hamiltonian positivity) to require that $\omega(v,Jv)\geq0$ for all $v\in W$. This makes $W$ to a (finite-dimensional) complex Hilbert space with the complex structure $J$.

A $K$-structure on an oriented 4-dim $\Sigma$ is a pseudo-conformal structure, and $\mc F\subset\mc X=\Omega(\Sigma,W)[2]$ is given by
\[\mc F=\bigl\{F\in\Omega^2(\Sigma,W) \mid *F=JF\bigr\}\oplus\Omega^{\geq3}(\Sigma,W). \qedhere\] 
\end{example}

\begin{example}[Yang-Mills]
\label{ex:ymbdry}
Let $\g$ be a Lie algebra  and let
$$X=T^*[n]T[1]\g[1]=\g^*[n-1]\times\g^*[n-2]\times\g[2]\times\g[1].$$

Consider an oriented $n$-dim vector space $V$ with an inner product. We search for an $SO(V)$-invariant dg Lagrangian submanifold
$$\mc F_V\subsetsep \mc X_V=\bw V^*\otimes X$$
satisfying the conditions of Proposition \ref{prop:F0}. For simplicity we shall suppose that $\mc F_V\subset \mc X_V$ is a graded vector subspace, and thus it needs to be a $\bw V^*$-submodule. Let us demand
$$\g[1]=\bw^0V^*\otimes\g[1]\subsetsep\mc F_V$$
(and thus $\bw V^*\otimes\g[1]\subset\mc F_V$) to impose gauge invariance (which implies, in particular, that $\mc F_V$ is not ghostless).

The generic such $\mc F_V$'s are of the following form: it is the sum of the dark fields in the table

\begin{figure}[h]
\vspace{-.5cm}
\centering
\begin{tikzpicture}
\matrix[table] (A) 
{
|[fill=mwhite]| & |[fill=mwhite]| & |[fill=mwhite]| & |[fill=mblack]|\\
|[fill=mwhite]| & |[fill=mwhite]| & |[fill=mwhite]| & |[fill=mblack]|\\
|[fill=mwhite]| & |[fill=mwhite]| & |[fill=mgray]| & |[fill=mblack]|\\
|[minimum height=\wigglycellheight, fill=mwhite]| & |[minimum height=\wigglycellheight, fill=mwhite]| & |[minimum height=\wigglycellheight, fill=mblack]| & |[minimum height=\wigglycellheight, fill=mblack]|\\
|[fill=mwhite]| & |[fill=mgray]| & |[fill=mblack]| & |[fill=mblack]|\\
|[fill=mwhite]| & |[fill=mblack]| & |[fill=mblack]| & |[fill=mblack]|\\
|[fill=mwhite]| & |[fill=mblack]| & |[fill=mblack]| & |[fill=mblack]|\\
};
\predpis{1}{$\bw^0 V^*$} \predpis{2}{$\bw^1 V^*$} \predpis{3}{$\bw^2 V^*$} \predpis{4}{$\dots$} \predpis{5}{$\bw^{n-2} V^*$} \predpis{6}{$\bw^{n-1} V^*$} \predpis{7}{$\bw^n V^*$}
\nadpis{1}{$\g^*[n-1]$} \nadpis{2}{$\g^*[n-2]$} \nadpis{3}{$\g[2]$} \nadpis{4}{$\g[1]$}
\draw[decorate,decoration={zigzag, segment length=\wigglyseglen,amplitude=\wigglyampl},  white, line width=\wigglywidth] ($(A-4-1) +(-1,0)$)--+(7.3,0);
\end{tikzpicture}
\vspace{-.5cm}
\end{figure}

{
}
\noindent plus the subspace of the sum of two gray fields 
$$
\bigl\{((*\otimes g)a,a)\mid a\in \bw^2 V^*\otimes\g[2]\bigr\}
$$
where $g:\g\to\g^*$ comes from an invariant inner product on $\g$.
\end{example}

\section{The field content of a typical AKSZ sandwich (degree counting)}\label{sec:htfiber}

Before passing to examples, let us describe the field (and ghost) content of a sandwich model given by $(X,L,\mc F)$ under  the following assumptions:
\begin{itemize}
\item $L\subset X$ is a dg Lagrangian submanifold (i.e.\  not  a more general Lagrangian map $L\to X$)
\item $\mc F$ is ultralocal and ghostless
\item $X$ is connected in the sense of rational homotopy theory \cite{Sul}, i.e.\ as a graded manifold $X$ is (isomorphic to) a \emph{negatively} graded vector space (this condition can be weakened without influencing the result)
\end{itemize}

Let us first describe the result.
Let $\ell:R\to X$ be a resolution of $L\subset X$ and let $\Phi:=\ell^{-1}(0)$ be its fiber (i.e.\ the homotopy fiber of $L\hookrightarrow X$). Then the physical fields of the sandwich are the grading-preserving maps
$$T[1]\Sigma\to\Phi.$$
In other words, if $\phi^i$ are local coordinates on $\Phi$ of degrees $d_i$ then the physical fields are forms 
$$A^i\in\Omega^{d_i}(\Sigma).$$

The ghosts (fields of negative degree) are then forms 
$$c^i_{[k]}\in\Omega^{d_i-k}(\Sigma),\quad 1\leq k\leq d_i$$
and the ghost number (i.e.\ minus the degree) of $c^i_{[k]}$ is $k$.

While, as usual in the BV formalism, it is impossible to disentangle the gauge symmetries from the equations of motion, the gauge symmetries are roughly speaking given by $\Phi$ seen as a higher Lie algebroid.

Let us now describe the local calculations leading to this result. Let $p_i$, $q^i$ be Darboux coordinates on $X$, with $L$ given by $q^i=0$. The resolution $R$ will then have local coordinates $\tilde p_i$, $\tilde q^i$, $\phi^i$ such that
$$\quad\omega_{\tilde Y}|^{\vphantom p}_{R}=d\tilde p_i\,d\phi^i$$
and such that 
$$\ell^*p_i=\tilde p_i + \text{higher},\quad\ell^*q^i=\tilde q^i +\text{higher}$$
 where ``higher'' means a function vanishing at least quadratically at the origin.
The degrees of the coordinates are 
$$\deg\phi^i=d_i,\ \deg \tilde q^i=d_i+1,\ \deg \tilde p_i=n-d_i-1,$$
 i.e.\ locally
$$R=\prod_i \R[d_i]\times\R[d_i+1]\times\R[n-d_i-1].$$
 As a result
$$\mc R=\on{Maps}(T[1]\Sigma,R)=\bigoplus_i\Omega(\Sigma)[d_i]\oplus\Omega(\Sigma)[d_i+1]\oplus\Omega(\Sigma)[n-d_i-1]$$
Let us denote the corresponding component forms by
$$\phi^{i(m)},\tilde q^{i(m)},\tilde p_i^{(m)}\in\Omega^m(\Sigma).$$

We now pass from $\mc R$ to $\mc Z=(\mc R_\mc F)^\text{reduced}$. After the restriction and reduction the following component forms remain:
$$\phi^{i(m)}\text{ for }m\leq d_i+1,\quad \tilde p_i^{(m)}\text{ for }m\geq n-d_i-1.$$

Let us conclude by giving the component forms the appropriate names:
\begin{align*}
\text{degree 0 (fields)}&:\quad A^i:=\phi^{i(d_i)},\quad B_i:=\tilde p_i^{(n-d_i-1)} \\
\text{degree $-k<0$ (ghosts)}&:\quad c^i_{[k]}:=\phi^{i(d_i-k)} \\
\text{degree > 0 (antifields)}&:\quad B^{+i}:=\phi^{i(d_i+1)},\quad A_i^+:=\tilde p_i^{(n-d_i)},\quad c^+_{[k]i}:=\tilde p_i^{(n-d_i+k)} 
\end{align*}
The fields $B_i$ should be considered as auxiliary, corresponding to the fact that $S_\mc Z$ is a first-order type action.

\section{Examples of sandwiches and dualities}\label{sec:ex}

\subsection{Electric-magnetic duality in higher dimensions}\label{sec:emhigh}
We start with the space $X=\R[a+1]\times\R[b+1]$, with coordinate $p$ on $\R[a+1]$ and $x$ on $\R[b+1]$. We take $Q_X=0$ and the degree $n=a+b+2$ symplectic form $\omega_X=dp\,dx$. This gives us
$$\mc X=\Omega(\Sigma)[a+1]\oplus\Omega(\Sigma)[b+1].$$

Supposing that $\Sigma$ has a (pseudo-)Riemannian metric, we define a ghostless ultralocal boundary condition $\mc F\subset\mc X$ via
\begin{multline*}
\mc F=\{(p,x)\in\Omega(\Sigma)[a+1]\oplus\Omega(\Sigma)[b+1] \mid\\
p^{(a+1)}=*x^{(b+1)},\, p^{(\leq a)}=0,\, x^{(\leq b)}=0\}
\end{multline*}
where $x^{(k)}$ denotes the $k$-form part of $x$. For $\theta_X=p\,dx$ we get
$$S_\mc F=\frac12\int_\Sigma x^{(b+1)}\,{*x^{(b+1)}}$$

The topological boundary condition will be given by $L=\R[b+1]\subset X$, i.e.\ by setting $p=0$. Let $R$ be the Koszul resolution of $L\subset X$, i.e.\ the coordinates on $R$ are $p$, $x$, and $\pi$ with $\deg\pi=a$, and the differential is $Q_R=p\partial_\pi$. The accompanying symplectic NQ manifolds are:
$$\tilde Y=\R[a]\times\R[b+1],\text{ coordinates $\pi$ and $x$,}\quad\omega_{\tilde Y}=dx\,d\pi,\ H_{\tilde Y}=0$$
$$Y=T^*[n]\tilde Y,\text{ coordinates $x$, $\pi$, $\xi$, $p$,}\quad\omega_Y=dp\,dx+d\xi\,d\pi,\ H_Y=\xi p$$
$$\text{coisotropic }R\subset Y\text{ given by }\xi=0$$
We have $\theta_Y^\text{taut}=p\,dx+\xi\,d\pi$. On $R$ it coincides with $\theta_X$ and we thus have $H_\text{rel}=0$.

The space of fields (i.e.\ the derived intersection) before the reduction is
\[\mc R_{\mc F}=\Omega(\Sigma)[a]\times\mc F\subsetsep\Omega(\Sigma)[a]\times\mc X=\mc R \]
with the additional component $\Omega(\Sigma)[a]$ corresponding to $\pi$.
Since $\omega_{\tilde Y}=dx\,d\pi$, the reduction of $\mc R_{\mc F}$ kills its $p$-component and a part of the $\pi$-component,  and we get
\[\mc Z\cong \Omega^{\le a+1}(\Sigma)[a]\oplus \Omega^{\ge b+1}(\Sigma)[b+1]\]
with the BV symplectic form pairing the two components.

Since $H_{\tilde Y}=0$ and $\omega_{\tilde Y}=dx\,d\pi$, we have
\[S_{\mathrm{AKSZ}}^{\tilde Y}=\int_\Sigma x\,d\pi.\]
Putting everything together, we obtain 
\[S_\mc Z=S_{\mathrm{AKSZ}}^{\tilde Y}+S_\mc F=\int_\Sigma x^{(n-1)}d\pi^{(0)}+\dots+x^{(b+1)}d\pi^{(a)}+\frac12x^{(b+1)}\,{*x^{(b+1)}}.\]

The appropriate names for the components are:
\begin{align*}
\text{degree 0 (fields)}&:\quad A:=\pi^{(a)},\quad B:=x^{(b+1)} \\
\text{degree $-k<0$ (ghosts)}&:\quad c_{[k]}:=\pi^{(a-k)},\quad k=1,\dots, a \\
\text{degree > 0 (antifields)}&:\quad B^+:=\pi^{(a+1)},\quad A^+:=x^{(b+2)},\quad c^+_{[k]}:=x^{(b+2+k)} 
\end{align*}
and the action now reads
$$
S_\mc Z=\int_\Sigma B\,dA+\frac12B\;{*B}+ A^+dc_{[1]}+c_{[1]}^+dc_{[2]}^{\vphantom{+}}+\dots+c_{[a-1]}^+dc_{[a]}^{\vphantom{+}}.
$$
This is the first order formulation of the action functional
$$\frac12\int_\Sigma dA\;{*dA}\qquad A\in\Omega^a(\Sigma)$$
(the first two terms of $S_\mc Z$) together with the ghost terms corresponding to the gauge transformations $A\to A+dA'$, $A'\to A'+dA''$, etc.

The dual model is obtained by considering the topological boundary condition given by $L'=\R[a+1]$, i.e.\ by exchanging $a$ and $b$.

\subsection{Electric-magnetic duality in 4d}\label{ss:em4}
Let us consider the theory arising from the Example \ref{ex:4dbdry}. Namely, we fix a symplectic vector space $W$ and we take $X=W[2]$. Together with $Q_X=0$, this makes the target $X$ an NQ symplectic manifold of degree $4$. We  have
$$\mc X=\Omega(\Sigma,W[2]).$$

The ghostless non-topological boundary condition, given by a pseudo-conformal structure on $\Sigma$ (which is a $K$-structure for the subgroup $K\subset GL_+(V)$ preserving a Minkowski inner product up to rescaling), is
$$\mc F=\{F\in\Omega^2(\Sigma,W[2])\mid *F=JF\}\oplus \Omega^{\geq3}(\Sigma,W[2])\subsetsep\mc X.$$
Here $J\colon W\to W$ a complex structure for which $\omega(\cdot,J\cdot)$ is symmetric and positive definite. Since $\mc F$ is $E_X$-invariant, taking $\theta_X=i_{E_X}\omega_X/4$ we have $\theta_\mc X|_\mc F=0$ and thus $S_\mc F=0$.

We decompose $W\cong U\oplus U^*$ into two Lagrangian vector subspaces and consider the topological boundary condition given by $L=U[2]$ ($U^*$ plays an auxiliary role). Let $R$  be the Koszul resolution of $U[2]\subset W[2]$. More explicitly, let $x^i$ denote the coordinates on $U$ and $p_i$ the corresponding dual coordinates on $U^*$. Then $R=W[2]\times U^*[1]$, with coordinates $x^i,p_i,\xi_i$,  and with the differential $Q_R=p_i\partial_{\xi_i}$. The corresponding acyclic $Y$ is 
\[Y=T^*[4]T^*[3]U^*[1]=W[2]\times T^*[4]U^*[1],\]
with the obvious embedding $R\hookrightarrow Y$, and $\tilde Y=T^*[3]U^*[1]=U[2]\times U^*[1]$ with $H_{\tilde Y}=0$.

The space of fields (i.e.\ the derived intersection) before the reduction is
\[\mc R_{\mc F}=\mc F\times \left(\Omega(\Sigma)\otimes U^*[1]\right).\]
For convenience, we have displayed (most of) this information in Figure \ref{im:em}, which is to be read as follows. 
The full table represents the space 
$$\mc R=\on{Maps}(T[1]\Sigma,R)= \Omega(\Sigma)\otimes R$$
 with cells corresponding to bihomogeneous components. The subspace $\mc {R_\mc F}\subset\mc R$ is the sum of the black cells plus a subspace of the sum of the gray cells
$$\bigl\{F\in\Omega^2\bigl(\Sigma,W[2]=U[2]\oplus U^*[2]\bigr) \mid *F=JF\bigr\}$$

\tikzset{ 
table/.style={
  matrix of nodes,
  row sep=-\pgflinewidth,
  column sep=-\pgflinewidth,
  nodes={rectangle,draw=black,text width=6ex,align=center},
  text depth=0.25ex,
  text height=3ex,
  nodes in empty cells
  },
texto/.style={font=\footnotesize\sffamily},
title/.style={font=\small\sffamily}
}
\definecolor{mblack}{gray}{0.25}
\definecolor{mgray}{gray}{0.7}
\definecolor{mwhite}{gray}{0.95}
\renewcommand\nadpis[2]{\node[anchor=base] at ([yshift=3.5ex]A-1-#1) {#2};}
\renewcommand\nadnadpis[2]{\node[anchor=base] at ([yshift=6.5ex]A-1-#1) {#2};}
\newcommand\janonadpis[2]{\node[anchor=base] at ([yshift=5ex]A-1-#1) {#2};}
\renewcommand\podpis[2]{\node[anchor=base] at ([yshift=-5.25ex]A-#1) {#2};}
\renewcommand\vpis[2]{\node[anchor=base] at ([yshift=-.8ex]A-#1) {#2};}
\renewcommand\vpisw[2]{\node[anchor=base] at ([yshift=-.8ex]A-#1) {\textcolor{white}{#2}};}
\renewcommand\predpis[2]{\node at ([xshift=-7.5ex]A-#1-1) {#2};}
\renewcommand\ocisluj[1]{\foreach[evaluate={\y=int(\x-1)}] \x in {1,...,#1} {\node at ([xshift=-6ex]A-\x-1) {$\y$};}}

\begin{figure}[h]
\centering
\begin{tikzpicture}
\matrix[table] (A) 
{
|[fill=mwhite]| & |[fill=mwhite]| & |[fill=mblack]|\\
|[fill=mwhite]| & |[fill=mwhite]|& |[fill=mblack]|\\
|[fill=mgray]| & |[fill=mgray]| & |[fill=mblack]|\\
|[fill=mblack]| & |[fill=mblack]| & |[fill=mblack]|\\
|[fill=mblack]| & |[fill=mblack]| & |[fill=mblack]|\\
};
\ocisluj{5}
\nadpis{1}{$U^*$} \nadpis{2}{$U$} \nadpis{3}{$U^*$}
\nadnadpis{1}{$-2$} \nadnadpis{2}{$-2$} \nadnadpis{3}{$-1$}
\podpis{5-1}{$p_i$} \podpis{5-2}{$x^i$} \podpis{5-3}{$\xi_i$}
\vpis{3-1}{$p^{(2)}$} \vpis{3-2}{$B$} \vpisw{1-3}{$c$} \vpisw{2-3}{$A$} \vpisw{3-3}{$B^+$} \vpisw{4-2}{$A^+$} \vpisw{5-2}{$c^+$}
\end{tikzpicture}
\vspace{-.2cm}
\caption{The space $\mc R_\mc F\subset \mc R$. The cells are labeled by components of $R$ (on horizontal axis) and by components of $\Omega(\Sigma)$ (on vertical axis). The name of the corresponding field is marked for some cells.}
\label{im:em}
\end{figure}

Since $\tilde\omega$ pairs the last two columns, the reduced space of fields is%
\[\mc Z\cong\Omega^{\ge 2}(\Sigma,U)[2]\oplus\Omega^{\le 2}(\Sigma,U^*)[1].\]
We then have, using the names from  Figure \ref{im:em}
\[S^{\tilde Y}_{\mathrm{AKSZ}}(p\circ f)=\int_\Sigma \la B,dA\ra+\la A^+, dc\ra. \]
For the last term in the action (\ref{eq:final}) we calculate
\[ H_\text{rel}=\tfrac{1}{4}i_{E_R}(j^*\theta^{\mathrm{taut}}_Y-\ell^*\theta_X)=\tfrac{1}{2}x^ip_i.\]
Expressing $p^{(2)}$ in terms of $B$ via the boundary condition $*(p^{(2)}+B)=J(p^{(2)}+B)$, we obtain
\[S_{\mc Z}=\int_\Sigma \la B,dA\ra+\la A^+, dc\ra + \alpha(B,B)+ \beta (B,*B),\]
for suitable $\alpha,\beta\in S^2U^*$.
This is the standard 1st order BV action of pure electrodynamics (with several charges), together with the topological term. 

Given $U$, any Lagrangian complement $U^*$ can be written as $$U^* = \{s(u) + J(u) \mid u\in U\},$$ where $s: U \to U$ is a linear map self-adjoint w.r.t.\ $g(u, v) = \omega(u, Jv)$. Then, ``integrating out'' $B$ (i.e.\ expressing it from its equation of motion and plugging the result back into the action), we get 
\[ S = \int_\Sigma \la dA , *dA \ra_g - \la s^t(dA), dA \ra_g\,. \]
Here, the pairing $\la\,,\,\ra_g$ on $U^*$ is the inverse of $g|_U$. 
Different choices of $U$ give different abelian Yang-Mills actions linked by S-duality \cite{EW}. Notice that the action $S$ depends on the choice of $U^*\subset W$ only  through the topological term $\la s^t(dA), dA \ra_g$.

 In the global picture involving gauge fields on non-trivial principal torus bundles we need to choose a lattice $\Lambda\subset W$, the vector subspaces $U,U^*\subset W$ should correspond to sub-tori of $W/\Lambda$ and the choice of $\Lambda$ must be such that $e^{iS}$ is independent of the choice of $U^*\subset W$.

%
%
%

\subsection{Scalar theory}\label{sec:scalar}

We consider $X=Y=T^*[n]\tilde Y$ for $\tilde Y=T^*[n-1]M$ and $M$ a (non-graded) manifold. Since $X$ is acyclic, we can take $R=X$, with the identity map $\ell\colon X\to X$ as the Lagrangian submersion. Choosing $\theta_X=\theta^{\mathrm{taut}}_X$, the formula (\ref{eq:final}) reduces to
\[S_{\mc R_\mc F}(f)=S_{\mathrm{AKSZ}}^{\tilde Y}(p\circ f)+S_\mc F(f).\]

Generic ghostless ultralocal boundary conditions $\mc F$ can be obtained as follows. Recall that $\mc F$ is completely determined by the Lagrangian submanifolds $\mc F^{\,0}_p:=\mc F_p\cap \mc X^{\,0}_p$, using the notation from Remark \ref{rem:negF}. 
We have
\[\mc X^{\,0}_p\cong T^*\mc P_p \otimes \bw^nT^*_p\Sigma,\]
where $\mc P_p$ is the space of grading preserving maps $T_p[1]\Sigma\to T^*[n-1]M$. 
In other words,
$$ \mc P_p:=T^*M\otimes \bw^{n-1}T^*_p\Sigma.$$
We choose a generating function $H_p\in C^\infty(\mc P_p)\otimes \bw^nT^*_p\Sigma$ and take $\mc F^{\,0}_p:=\on{graph}(dH_p)$.

\begin{figure}[h]
\centering
\begin{tikzpicture}
\matrix[table] (A) 
{
|[fill=mwhite]| & |[fill=mwhite]| & |[fill=mwhite]| & |[fill=mgray]|\\
|[fill=mwhite]| & |[fill=mwhite]| & |[fill=mgray]| & |[fill=mblack]|\\
|[minimum height=\wigglycellheight, fill=mwhite]| & |[minimum height=\wigglycellheight, fill=mwhite]| & |[minimum height=\wigglycellheight, fill=mblack]| & |[minimum height=\wigglycellheight, fill=mblack]|\\
|[fill=mwhite]| & |[fill=mgray]| & |[fill=mblack]| & |[fill=mblack]|\\
|[fill=mgray]| & |[fill=mblack]| & |[fill=mblack]| & |[fill=mblack]|\\
};
\predpis{1}{0} \predpis{2}{1} \predpis{3}{$\dots$} \predpis{4}{$n-1$} \predpis{5}{$n$}
\nadpis{1}{$T^*M$} \nadpis{2}{$T^*M$} \nadpis{3}{$TM$} \nadpis{4}{$M$}
\nadnadpis{1}{$-n$} \nadnadpis{2}{$1-n$} \nadnadpis{3}{$-1$} \nadnadpis{4}{0}
\vpis{1-4}{$x$} \vpisw{2-4}{$\pi^+$} \vpis{4-2}{$\pi$} \vpisw{5-2}{$x^+$}
\draw[decorate,decoration={zigzag, segment length=\wigglyseglen,amplitude=\wigglyampl},  white, line width=\wigglywidth] ($(A-3-1) +(-0.8,0)$)--+(5.2,0);
\end{tikzpicture}
\vspace{-.2cm}
\caption{The space $(\mc R_\mc F)_p\subset {\mc R}_p$ for the scalar theory. In order to understand the table as a graded vector space, we here identify locally $M\cong \R^m$.}
\end{figure}


Denoting the space of grading preserving maps $T[1]\Sigma\to T^*[n-1]M$ by $\mc P$, we get $\mc Z\cong T^*[-1]\mc P$. The space $\mc P$ corresponds to the fields and the fibers of $T^*[-1]$ to the antifields.

Choosing coordinates on $M$, we can describe $\mc Z$ via a tuple of forms $(x,\pi,x^+,\pi^+)$, where 
\[ x\in \Omega^0(\Sigma, \R^n),\,\, \pi\in \Omega^{n-1}(\Sigma, (\R^n)^*), \,\, x^+\!\in \Omega^n(\Sigma, (\R^n)^*),\,\, \pi^+\!\in \Omega^{1}(\Sigma, \R^n).\]
We can now write the action as
\[S_{\mc Z}=\int_\Sigma \pi_i dx^i+ H_p(\pi,x).\]
In particular, we see that the antifields $x^+,\pi^+$ do not enter in $S_{\mc Z}$.

Integrating out $\pi$'s we obtain an action of the form\footnote{This corresponds to a partial Legendre transform.}
\[S=\int_\Sigma \mathscr{L}(x,dx).\]
Note that if we call the space of grading preserving maps $T_p[1]\Sigma\to T[1]M$ by $\mc W_p$, we can interpret $\mathscr{L}$ (at $p\in \Sigma$) as the generating function for $\mc F^{\,0}_p$ in $\mc X^{\,0}_p\cong T^*\mc W_p \otimes \bw^nT^*_p\Sigma$.

\subsection{Yang-Mills}

Let $\g$ be a Lie algebra and let
$$X=T^*[n]T^*[n-1]\g[1]=\g^*[n-1]\times\g^*[n-2]\times\g[2]\times\g[1].$$
Since $X$ is acyclic, we can again take $Y=R=X$ (with $\ell=\mathrm{id}_X$) and set $\theta_X=\theta^{\mathrm{taut}}_X$.

We take $\mc F$ from Example \ref{ex:ymbdry} as our non-topological boundary condition. Recall that this is given by an invariant inner product $\la\cdot,\cdot\ra_\g$ on $\g$ and a pseudo-Riemannian metric on $\Sigma$. 

Choosing a basis of $\g$ we get linear coordinates $\theta^i$ on $\g[1]$, $t^i$ on $\g[2]$, $\sigma_i$ on $\g^*[n-2]$, and $s_i$ on $\g^*[n-1]$. 
We have 
$$S_\mc F=\frac12\int_{T[1]\Sigma}\la \sigma^{(n-2)}*\sigma^{(n-2)}\ra_\g.$$

\begin{figure}[h]
\centering
\begin{tikzpicture}
\matrix[table] (A) 
{
|[fill=mwhite]| & |[fill=mwhite]| & |[fill=mwhite]| & |[fill=mblack]|\\
|[fill=mwhite]| & |[fill=mwhite]| & |[fill=mwhite]| & |[fill=mblack]|\\
|[fill=mwhite]| & |[fill=mwhite]| & |[fill=mgray]| & |[fill=mblack]|\\
|[minimum height=\wigglycellheight, fill=mwhite]| & |[minimum height=\wigglycellheight, fill=mwhite]| & |[minimum height=\wigglycellheight, fill=mblack]| & |[minimum height=\wigglycellheight, fill=mblack]|\\
|[fill=mwhite]| & |[fill=mgray]| & |[fill=mblack]| & |[fill=mblack]|\\
|[fill=mwhite]| & |[fill=mblack]| & |[fill=mblack]| & |[fill=mblack]|\\
|[fill=mwhite]| & |[fill=mblack]| & |[fill=mblack]| & |[fill=mblack]|\\
};
\predpis{1}{0} \predpis{2}{1} \predpis{3}{2} \predpis{4}{$\dots$} \predpis{5}{$n-2$} \predpis{6}{$n-1$} \predpis{7}{$n$}
\nadpis{1}{$\g^*$} \nadpis{2}{$\g^*$} \nadpis{3}{$\g$} \nadpis{4}{$\g$}
\nadnadpis{1}{$1-n$} \nadnadpis{2}{$2-n$} \nadnadpis{3}{$-2$} \nadnadpis{4}{$-1$}
\podpis{7-1}{$s_i$} \podpis{7-2}{$\sigma_i$} \podpis{7-3}{$t^i$} \podpis{7-4}{$\theta^i$}
\vpisw{1-4}{$c$} \vpisw{2-4}{$A$} \vpisw{3-4}{$B^+$} \vpis{5-2}{$B$} \vpisw{6-2}{$A^+$} \vpisw{7-2}{$c^+$}
\draw[decorate,decoration={zigzag, segment length=\wigglyseglen,amplitude=\wigglyampl},  white, line width=\wigglywidth] ($(A-4-1) +(-0.8,0)$)--+(5.0,0);
\end{tikzpicture}
\vspace{-.2cm}
\caption{The space $\mc R_\mc F\subset \mc R$ for the Yang-Mills case. The degree $-1$ symplectic form $\tilde{\omega}_\mc R$ pairs columns $\sigma_i$ and $\theta^i$.}
\label{im:ym}
\end{figure}


As can be seen from Figure \ref{im:ym}, the reduced space is
\[\mc Z\cong \left(\Omega^{\ge n-2}(\Sigma)\otimes\g^*[n-2]\right)\oplus \left(\Omega^{\le 2}(\Sigma)\otimes\g[1]\right).\]
Since $\tilde Y = T^*[n-1]\g[1]$, we get $\theta_{\tilde Y}=\sigma_i\,d\theta^i$ and $H_{\tilde Y}=\tfrac12 f_{ij}^k \sigma_k\theta^i\theta^j$, and the action is
\[S_{\mc Z}=\int_\Sigma \sigma_i\,d\theta^i+\tfrac12 f_{ij}^k \sigma_k\theta^i\theta^j+\tfrac12\la \sigma^{(n-2)}*\sigma^{(n-2)}\ra_\g,\]
where $\sigma_i$, $\theta^i$ are now understood as inhomogeneous differential forms on $\Sigma$. 

Renaming the variables as in Figure \ref{im:ym} and writing $g^{ij}$ for the  inverse of the matrix of~$\la\cdot,\cdot\ra_\g$, we get
\begin{align*}
S_{\mc Z}=\int_\Sigma B_idA^i + \tfrac12 f_{ij}^k &B_kA^iA^j + \tfrac12g^{ij}B_i *B_j\\
& + A_i^+dc^i + f_{ij}^k A^+_kc^iA^j + f_{ij}^k B_kc^iB^{+j} + \tfrac12 f_{ij}^k c^+_kc^ic^j.
\end{align*}

This is the Yang-Mills action in the first order BV formulation. Taking only the part with fields of degree 0 (i.e.\ removing the terms containing ghosts and antifields) we get
\[S^{(0)}_{\mc Z}=\int_\Sigma B_idA^i+\tfrac12 f_{ij}^k B_kA^iA^j+\tfrac12g^{ij}B_i\, {*} B_j\]
and integrating out the auxiliary field $B$,
we obtain
\[S_{\textrm{Yang-Mills}}=\frac12\int_\Sigma \la F*F\ra_\g,\]
for $F^k=dA^k+\tfrac12 f_{ij}^k A^iA^j$.

%
\subsection{Poisson-Lie T-duality}\label{sec:pl}
Let us consider the case $n=2$, with $X=\g[1]$ for some quadratic Lie algebra $\g$. In this case, the corresponding AKSZ-model is the Chern-Simons theory.

A ghostless non-topological boundary condition $\mc F$ will be obtained as in Example \ref{ex:GM} out of a pseudo-conformal structure on $\Sigma$ and a generalized metric $E:\g\to\g$. We set
$$\mc F=\{A\in\Omega^1(\Sigma,\g)\mid *A=EA\}\oplus\Omega^2(\Sigma,\g).$$

For a topological boundary condition, we take $L=\h[1]$ where $\h\subset\g$ is a Lagrangian Lie subalgebra. Its resolution is given in Example \ref{ex:G/H}, i.e.\
$$R=\g[1]\times G/H \subsetsep \g[1]\times T^*[2](G/H)=Y.$$

In the language of Courant algebroids, $Y$ corresponds to the exact Courant algebroid $\g\times G/H\to G/H$, with the bracket and pairing of constant sections given by $\g$ and with the anchor map given  by the action of $\g$ on $G/H$. 

To proceed, we need to put $Y$ to the form
$$Y\cong T^*[2]\tilde Y,\quad \tilde Y=T^*[1](G/H),$$
or equivalently, we need an isomorphism of the exact Courant algebroid $\g\times G/H$ with the standard Courant algebroid $(T\oplus T^*)(G/H)$. In general this can be done only locally, as the class of $\g\times G/H$ in $H^3(G/H;\R)$ may be non-zero.

For simplicity, let us suppose in the fashion of Example \ref{ex:gradedplgrps} that $\g$ can be decomposed (as a vector space) into a direct sum of two complementary Lagrangian Lie subalgebras $\g=\h\oplus\h'$ and take $H'$ to be a Poisson-Lie group corresponding to $\h'$. The local isomorphism $H'\cong G/H$ then gives the (local) identification of NQ manifolds
\begin{equation}
\label{eq:action_vs_poisson}
\h[1]\times H'\cong T^*[1]H',
\end{equation}
where the differential on the LHS is given by the action of $\h$ on $H'\cong G/H$, on the RHS it comes from the Poisson bivector $\pi$ on the Poisson-Lie group $H'$, and the identification is given by the left trivialization.
We can thus write 
\[Y\cong\g[1]\times T^*[2]H'\cong T^*[2](\h[1]\times H')\cong T^*[2]T^*[1]H',\]
\[R=\g[1]\times H', \qquad \tilde Y=T^*[1]H',\qquad X\cong T^*[2]\h[1]=\h'[1]\times\h[1].\]
Choosing $\theta_X$ to be the tautological 1-form on $T^*[2]\h[1]$ we get
\[j^*\theta^{\mathrm{taut}}_Y-\ell^*\theta_X=0\]
and thus $H_{rel}=0$.

The reduced space is
\[\mc Z\cong T^*[-1]\bigl(\on{Maps}(\Sigma,H')\times\Omega^1(\Sigma,\h)\bigr)\]
as in \S\ref{sec:scalar} (with $M=H'$).
Let us use the notation $g\in \on{Maps}(\Sigma,H')$ and $B\in\Omega^1(\Sigma,\h)$ for the fields of degree 0. Note that for degree reasons, these will be the only fields entering into the action $S_\mc Z$.

\begin{figure}[h]
\centering
\begin{tikzpicture}
\matrix[table] (A) 
{
|[fill=mwhite]| & |[fill=mwhite]| & |[fill=mblack]|\\
|[fill=mgray]| & |[fill=mgray]| & |[fill=mblack]|\\
|[fill=mblack]| & |[fill=mblack]| & |[fill=mblack]|\\
};
\ocisluj{3}
\nadpis{1}{$\h'$} \nadpis{2}{$\h$} \nadpis{3}{$H'$}
\nadnadpis{1}{$-1$} \nadnadpis{2}{$-1$} \nadnadpis{3}{$0$}
\vpisw{1-3}{$g$} \vpis{2-2}{$B$}
\end{tikzpicture}
\vspace{-.2cm}
\caption{The space $\mc R_\mc F\subset \mc R$ for the case of Poisson-Lie T-duality. The last column is described using a local chart on $H'$ (so that the table represents a graded vector space).}
\label{im:plt}
\end{figure}


To compute $S_\mc F$ we use the Lagrangian splitting
\[\mc X_V^{\,0}=\on{Hom}(V,\h)\oplus \on{Hom}(V,\h')\]
and (supposing $\mc {F}_V^{\,0}\cap \on{Hom}(V,\h')=0$) understand $\mc {F}_V^{\,0}$ as the graph of a linear map 
\[\on{Hom}(V,\h)\to \on{Hom}(V,\h')\cong \on{Hom}(V,\h)^*\otimes\bw^2 V^*,\]
corresponding to a bilinear map $\psi\colon \on{Hom}(V,\h)^{\otimes 2}\to\bw^2 V^*$.
For the above choice of $\theta_X$ we then have
\[S_\mc F(p\circ f)=\frac12\int_\Sigma\psi (B,B).\]


The full action is 
\[S_\mc Z=\int_\Sigma \la B,g^{-1}dg\ra+\tfrac12(g^{-1}\pi_g)(B,B)+\tfrac12\psi(B,B).\]
Here $\pi_g$ is the Poisson bivector at $g\in H'$ and $g^{-1}\pi_g\in\bw^2\h'$ its left translate to the origin.
Integrating out $B$ we get
\[S=\frac12\int_\Sigma e(g^{-1}dg,g^{-1}dg),\quad e_g=(g^{-1}\pi_g+\psi)^{-1}.\]
The Poisson-Lie T-duality \cite{KS} corresponds to the switching of the roles of $\h$ and $\h'$, or to a choice of a different pair $\h,\h'\subset\g$.

\begin{rem}
The more general boundary condition from Example \ref{ex:dressing} gives rise to the BV picture of the Poisson-Lie T-duality for gauged $\sigma$-models (dressing cosets) from \cite{KS2}.
\end{rem}

\subsection{Higher Poisson-Lie T-duality}
Let us now generalize the previous example to higher dimensions. Suppose $\g$ is a graded Lie algebra concentrated in non-positive degrees, with an invariant pairing of degree $n-2$ so that 
$$X=\g[1]$$
 is an NQ symplectic manifold with $\omega_X$ of degree $n$.
Let $\h\subset\g$ be a graded Lagrangian Lie subalgebra, and let us set $L=\h[1]\subset X$. This gives us, as in Example \ref{ex:G/H}, the resolution $R=\g[1]\times G/H$, and the fiber $\Phi$ from \S\ref{sec:htfiber} is the NQ-manifold $\Phi=G/H$ (with $Q_\Phi=0$). The examples from \S\ref{sec:emhigh}, \S\ref{ss:em4} and \S\ref{sec:pl} are special cases of this setup.

To compute the action functional using our methods, let us suppose that $\h'\subset\g$ is another graded Lagrangian Lie subalgebra, which is complementary to $\h$ (which gives us $\h'\cong \h^*[n-2]$).
  Example \ref{ex:gradedplgrps} gives us a  resolution of $\h[1]\subset\g[1]$ as
\[R=\g[1]\times H'\subsetsep T^*[n](\h[1]\times H')\cong T^*[n]T^*[n-1]H'=Y,\]
where $H'$ is a graded Poisson-Lie group integrating $\h'$.

Let us now choose a ghostless $\mc F\subset\mc X=\Omega(\Sigma,\g[1])$, or equivalently a Lagrangian submanifold $\mc F^0\subset \mc X^0$. A simple way to do it is to notice that
$$\mc X=\Omega(\Sigma,\h[1])\oplus\Omega(\Sigma,\h'[1])\cong T^*\Omega(\Sigma,\h[1])$$
and similarly
$$\mc X^0=T^*\Omega(\Sigma,\h[1])^0.$$
We thus choose a functional $S_\mc F$ on $\Omega(\Sigma,\h[1])^0$ and use is as a generating function of $\mc F^0$. The resulting space $\mc R_\mc F$ is depicted in the table. (The most natural case is when $S_\mc F$ contains no derivatives of the fields; this gives rise to an ultralocal $\mc F$.)
 

\begin{figure}[h]
	\centering
	\begin{tikzpicture}
	\matrix[table] (A) 
	{
		|[fill=mwhite]| & |[fill=mwhite]| & |[fill=mwhite]| &[1.5mm] |[fill=mwhite]| & |[fill=mwhite]| & |[fill=mwhite]| &[1.5mm] |[fill=mblack]| & |[fill=mblack]| & |[fill=mblack]|\\
		|[fill=mwhite]| & |[fill=mwhite]| & |[fill=mgray]| & |[fill=mwhite]| & |[fill=mwhite]| & |[fill=mgray]| & |[fill=mblack]| & |[fill=mblack]| & |[fill=mblack]|\\
		|[minimum height=\wigglycellheight, fill=mwhite]| & |[minimum height=\wigglycellheight]| & |[minimum height=\wigglycellheight, fill=mblack]| & |[minimum height=\wigglycellheight, fill=mwhite]| & |[minimum height=\wigglycellheight]| & |[minimum height=\wigglycellheight, fill=mblack]| & |[minimum height=\wigglycellheight, fill=mblack]| & |[minimum height=\wigglycellheight, fill=mblack]| & |[minimum height=\wigglycellheight, fill=mblack]|\\
		|[fill=mgray]| & |[fill=mblack]| & |[fill=mblack]| & |[fill=mgray]| & |[fill=mblack]| & |[fill=mblack]| & |[fill=mblack]| & |[fill=mblack]| & |[fill=mblack]|\\
		|[fill=mblack]| & |[fill=mblack]| & |[fill=mblack]| & |[fill=mblack]| & |[fill=mblack]| & |[fill=mblack]| & |[fill=mblack]| & |[fill=mblack]| & |[fill=mblack]|\\
	};
	\filldraw[draw=black, fill=mblack]  ($(A-3-2.south east)+(-0.15pt,0.15pt)$)rectangle(A-3-2);
	\filldraw[draw=black, fill=mgray]  ($(A-3-2.south west)+(0.15pt,0.15pt)$)rectangle(A-3-2);
	\filldraw[draw=black, fill=mgray]  ($(A-3-2.north east)+(-0.15pt,-0.15pt)$)rectangle(A-3-2);
	\filldraw[draw=black, fill=mwhite]  ($(A-3-2.north west)+(0.15pt,-0.15pt)$)rectangle(A-3-2);
	\filldraw[draw=black, fill=mblack]  ($(A-3-5.south east)+(-0.15pt,0.15pt)$)rectangle(A-3-5);
	\filldraw[draw=black, fill=mgray]  ($(A-3-5.south west)+(0.15pt,0.15pt)$)rectangle(A-3-5);
	\filldraw[draw=black, fill=mgray]  ($(A-3-5.north east)+(-0.15pt,-0.15pt)$)rectangle(A-3-5);
	\filldraw[draw=black, fill=mwhite]  ($(A-3-5.north west)+(0.15pt,-0.15pt)$)rectangle(A-3-5);
	\draw[decorate,decoration={zigzag, segment length=\wigglyseglen,amplitude=\wigglyampl},  white, line width=\wigglywidth] ($(A-1-2) +(0,0.395)$)--+(0,-4.3);
\draw[decorate,decoration={zigzag, segment length=\wigglyseglen,amplitude=\wigglyampl},  white, line width=\wigglywidth] ($(A-1-5) +(0,0.395)$)--+(0,-4.3);
\draw[decorate,decoration={zigzag, segment length=\wigglyseglen,amplitude=\wigglyampl},  white, line width=\wigglywidth] ($(A-1-8) +(0,0.395)$)--+(0,-4.3);
	\predpis{1}{0} \predpis{2}{1} \predpis{3}{$\dots$} \predpis{4}{$n-1$} \predpis{5}{$n$}
	\nadpis{1}{$\h'$} \janonadpis{2}{$\cdots$} \nadpis{3}{$\h'$} \nadpis{4}{$\h$}  \nadpis{6}{$\h$} \nadpis{7}{$H'$}  \nadpis{9}{$H'$}
	\nadnadpis{1}{$1-n$}  \nadnadpis{3}{$-1$} \nadnadpis{4}{$1-n$} \janonadpis{5}{$\cdots$} \nadnadpis{6}{$-1$} \nadnadpis{7}{$2-n$} \janonadpis{8}{$\cdots$} \nadnadpis{9}{$0$}
	\draw[decorate,decoration={zigzag, segment length=\wigglyseglen,amplitude=\wigglyampl},  white, line width=\wigglywidth] ($(A-3-1) +(-0.7,0)$)--+(10.9,0);
	\end{tikzpicture}
	\vspace{-.2cm}
	\caption{$\mc R_\mc F\subset \mc R$ for the higher Poisson-Lie T-duality.}
\end{figure}


To describe the quotient space $\mc Z$ of $\mc R_\mc F$ we now take $\mc H'_2$ to be the graded normal subgroup of the graded group $\on{Maps}(T[1]\Sigma,H')$ with the Lie algebra $\Omega(\Sigma, \h')^{\geq2}$. We then have
\[\mc Z\cong\Omega(\Sigma,\h[1])^{\ge 0}\times \on{Maps}(T[1]\Sigma, H')/{\mc H'_2}.\]
Denoting by $\bar B$ and $\bar g$ the two components we then have
\[S_\mc Z(\bar B,\bar g)=\int_\Sigma \la \bar B,\bar g^{-1}d\bar g\ra + \tfrac12(\bar g^{-1}\pi_{\bar g})(\bar B,\bar B)+S_\mc F(B).\]
Here $ B\in\Omega(\Sigma,\h[1])^0$ is the degree-0 component of $\bar B$, we understand $\bar g^{-1}\pi_{\bar g}$ as a function on $H'$ with values in $\h'\otimes\h'$, and $\bar g^{-1}d\bar g=i_{d_{\Sigma}}\bar g^*\theta_{MC}$, where $\theta_{MC}$ is the Maurer-Cartan form on $H'$.
%
%

The action is written in the proper BV form, including all the ghosts corresponding to its respective gauge symmetries. The dual action is then obtained by reversing the role of $\h$ and $\h'$, or possibly by choosing another pair $\h,\h'\subset\g$. 
\medskip

For $n=3$, the general case consists of the semi-abelian double $\g=\mathfrak{k}\oplus\mathfrak{k}^*[1]$ for $\mathfrak{k}$ a Lie algebra. Suppose we decompose $\kk$ into a direct sum (in the vector space sense) of two complementary Lie subalgebras $\kk=\kk_1\oplus\kk_2$ (of arbitrary dimensions). We then set $\h=\kk_1\oplus \on{Ann}\kk_1[1]$ and $\h'=\kk_2\oplus \on{Ann}\kk_2[1]$.

\begin{example} 
For illustration, let us take $\h=\mathfrak{k}^*[1]$, and $\h'=\mathfrak{k}$, and let $K$ be a Lie group corresponding to $\kk$.

\begin{figure}[h]
\begin{minipage}{0.45\hsize}
\centering
\setlength\extrarowheight{1pt}

\begin{tikzpicture}
\matrix[table] (A) 
{
|[fill=mwhite]| & |[fill=mwhite]| & |[fill=mblack]|\\
|[fill=mgray]| & |[fill=mwhite]| & |[fill=mblack]|\\
|[fill=mblack]| & |[fill=mgray]| & |[fill=mblack]|\\
|[fill=mblack]| & |[fill=mblack]| & |[fill=mblack]|\\
};
\ocisluj{4}
\nadpis{1}{$\mathfrak{k}$} \nadpis{2}{$\mathfrak{k}^*$} \nadpis{3}{$K$}
\nadnadpis{1}{$-1$} \nadnadpis{2}{$-2$} \nadnadpis{3}{$0$}
\vpisw{1-3}{$ g$} \vpisw{2-3}{$ B^+$} \vpis{3-2}{$ B$} \vpisw{4-2}{$ g^+$}
\end{tikzpicture}
\vspace{-.2cm}
\caption{$\mc R_\mc F\subset \mc R$ for boundary condition given by $\h$.}
\end{minipage} \begin{minipage}{0.45\hsize}
\centering
\setlength\extrarowheight{1pt}
\begin{tikzpicture}
\matrix[table] (A) 
{
|[fill=mwhite]| & |[fill=mwhite]| & |[fill=mblack]|\\
|[fill=mwhite]| & |[fill=mgray]| & |[fill=mblack]|\\
|[fill=mgray]| & |[fill=mblack]| & |[fill=mblack]|\\
|[fill=mblack]| & |[fill=mblack]| & |[fill=mblack]|\\
};
\ocisluj{4}
\nadpis{1}{$\mathfrak{k}^*$} \nadpis{2}{$\mathfrak{k}$} \nadpis{3}{$\mathfrak{k}^*$}
\nadnadpis{1}{$-2$} \nadnadpis{2}{$-1$} \nadnadpis{3}{$-1$}
\vpisw{1-3}{$c$} \vpisw{2-3}{$A$} \vpisw{3-3}{$B^+$} \vpis{2-2}{$B$} \vpisw{3-2}{$A^+$} \vpisw{4-2}{$c^+$}
\end{tikzpicture}
\vspace{-.2cm}
\caption{$\mc R_\mc F\subset \mc R$ for boundary condition given by $\h'$.}
\label{im:second_illustration}
\end{minipage}
\end{figure}

We first consider the boundary condition given by $\h$. We have $H'=K$, with the zero Poisson bivector. The reduced space is
\[\mc Z\cong T^*[-1]\left(\on{Maps}(\Sigma,K)\times \Omega^2(\Sigma,\mathfrak{k}^*)\right).\]
Using the notation $ g\in\on{Maps}(\Sigma,K)$, $ B\in\Omega^2(\Sigma,\mathfrak{k}^*)$, the action becomes
\[S_\h=\int_\Sigma \la  B, { g}^{-1}d g\ra+S_\mc F( B).\]
The field $B$ should be understood as auxiliary; once we integrate it out we get an action functional for maps $g:\Sigma\to K$.


We now take the boundary condition determined by $\h'$. Integrating $\h$ we obtain the abelian group $H=\mathfrak{k}^*[1]$, with the Kirillov-Kostant-Souriau Poisson structure.
Using the notation from Figure \ref{im:second_illustration}, we obtain the action
\begin{align*}
S_{\h'}&=\int_\Sigma \la B,dA\ra +\la A^+,dc\ra+\tfrac12 \la A,[B,B]\ra+\la c,[B,A^+]\ra+S'_\mc F(B)\\
&=\int_\Sigma \la B,D_BA\ra+\la A^+,D_B c\ra+S'_\mc F(B),
\end{align*}
where $D_B$ is the covariant derivative w.r.t.\ the connection $B\in\Omega^1(\Sigma,\kk)$ and $S'_\mc F$ is obtained from $S_\mc F$ via the Legendre transform.

The field $c$ is a ghost. The corresponding gauge group is the abelian group $\on{Maps}(\Sigma,\kk^*)$, acting via $A\mapsto A+D_Bu$, where $u\in \on{Maps}(\Sigma,\kk^*)$.
\end{example}

The general $n=4$ case has the form $\g=\kk\oplus W[1]\oplus \kk^*[2]$, with $W$ a symplectic vector space, and the graded Lie bracket given by a Lie algebra structure on $\kk$ and a symplectic representation of $\kk$ in $W$. Lagrangian Lie subalgebras are of the form $\mf n\oplus U[1]\oplus \on{Ann}\mf n[2]$, with $\mf n\subset\kk$ a Lie subalgebra, and $U\subset W$ an $\mf n$-invariant Lagrangian subspace. 

As mentioned above, the homotopy fiber  of $L\subset X$ is $\Phi=G/H$. The physical fields of the resulting theory (after eliminating the auxiliary fields) are thus maps to $K/N$,  1-forms valued in  $U^*\cong W/U$ (c.f.\ \S\ref{ss:em4}), and 2-forms valued in $\mf n^*$ (or rather valued in the vector bundles over $K/N$ associated to these two representations of $N$). Furthermore, there is a set of ghosts for the $U^*$-valued 1-forms, and ghosts together with ghosts for ghosts for the $\mf n^*$-valued 2-forms. 
\medskip

It should be noted that for dimensions $n\geq3$ the ansatz $X=\g[1]$, with $\g$ a graded Lie algebra,  is somewhat restrictive. One should allow $\g$ to be a minimal $L_\infty$-algebra, i.e.\ allow the Hamiltonian $H_X$ to have also quartic and higher terms. The resulting homotopy fibre $\Phi$ would then still be $G/H$ (where $G$ and $H$ are the N-groups integrating $\g$ and $\h$ when we keep only the binary bracket, i.e.\ only the cubic part of $H_X$), but now $Q_\Phi$ might be non-zero.

\end{document}